\documentclass[conference]{IEEEtran}
\IEEEoverridecommandlockouts
\usepackage{cite}
\usepackage{amsmath,amssymb,amsfonts,amsthm}

\newtheorem{theorem}{Theorem}%  meant for continuous numbers
\usepackage{graphicx,float}
\usepackage{textcomp}
\usepackage{xcolor}
\def\BibTeX{{\rm B\kern-.05em{\sc i\kern-.025em b}\kern-.08em
    T\kern-.1667em\lower.7ex\hbox{E}\kern-.125emX}}
\begin{document}

\title{\LARGE \bf
Control of Vehicle Platoons with Collision Avoidance Using Noncooperative Differential Games*}

\author{Hossein B. Jond$^{1}$% <-this % stops a space
\thanks{*This work was supported by SGS, V\v{S}B - Technical University of Ostrava, Czech Republic, under grant No. SP2023/012 “Parallel processing of Big Data X”.}% <-this % stops a space
\thanks{$^{1}$Hossein B. Jond is with the Department of Computer Science, V\v{S}B-Technical University of Ostrava, 708 00 Ostrava-Poruba, Czech Republic
        {\tt\small hossein.barghi.jond@vsb.cz}}%
}

\maketitle

\begin{abstract}
This paper considers a differential game approach to the predecessor-following vehicle platoon control problem without and with collision avoidance. In this approach, each vehicle tries to minimize the performance index (PI) of its control objective, which is reaching consensual velocity with the predecessor vehicle while maintaining a small inter-vehicle distance from it. Two differential games were formulated. The differential game problem for platoon control without collision avoidance is solved for the open-loop Nash equilibrium and its associated state trajectories. The second differential game problem for platoon control with collision avoidance has a non-quadratic PI, which poses a greater challenge to obtaining its open-loop Nash equilibrium. Since the exact solution is unavailable, we propose an estimated Nash strategy approach that is greatly simplified for implementation. An illustrative example of a vehicle platoon control problem was solved under both the without and with collision avoidance scenarios. The results showed the effectiveness of the models and their solutions for both scenarios.
\end{abstract}


\section{Introduction}
Convoy and platoon group driving are the salient collective behaviors of connected and automated vehicles on the road~\cite{Bergenhem2012,Soni2018}. Vehicles in a platoon or convoy drive at a consensual speed in the direction of the flow of traffic while maintaining a small inter-vehicle distance from their adjacent vehicles. In a platoon, we are concerned only with the longitudinally coordinated control of vehicles moving in the same lane of the road or highway~\cite{Guo2011,JondIET2023}. In a convoy, both longitudinal and lateral coordination of vehicles over different lanes is necessary~\cite{QianFM16,JondIEEE2023}. 

Platooning is the most studied group behavior of connected and automated vehicles~\cite{WOO2021103442,Lesch2022}. Such coordination is achieved by exchanging local information among the vehicles~\cite{Abualhoul2013}. Vehicle platoons offer remarkable benefits, as listed in~\cite{Bergenhem2012,Wang2020,Lesch2022}. Platooning boosts road capacity and decreases fuel consumption and emissions of pollutants due to the decrease in gaps between vehicles and the elimination of dispensable changes in speed and aerodynamic drag on the following vehicles, respectively. Besides, driving safety and passenger satisfaction are enhanced since detection and actuation times are shorter, and the small inter-vehicle gaps between vehicles prevent cut-ins by other vehicles. The most common platooning methods are the Leader-Follower approach~\cite{Dunbar2012}, the Behavior-Based Approach~\cite{FIROOZI2021104714,Antonelli2010}, and the Virtual Structure approach~\cite{ELZAHER20121503,Ren2004}. Several other methods were also researched~\cite{Semsar-Kazerooni,Hao2023,Li2018-01-1646}.

In the classical optimal control framework, vehicles attempt to acquire a platoon formation by optimizing a team objective~\cite{Morbidi2013}. This framework, however, is incompatible with automated and autonomous vehicles, which are supposed to make independent and selfish operational decisions without the need for human intervention. Game theory provides tools and concepts for determining the best strategy or action choices for each vehicle with self-interests and acting selfishly. A vehicle’s interest in a platoon could be to penalize its relative displacement, velocity, and acceleration errors, taking its fuel amount into account~\cite{LIN201454}. The strategic interactions among vehicles acting independently and selfishly naturally portray a noncooperative game. Nash equilibrium allows for self-enforcing strategic interactions in a noncooperative game~\cite{van2012refinements}. Platooning emerges as a result of a Nash equilibrium~\cite{JondIET2023}.

Game-theoretic platoon control has recently attracted increasing interest in the control community. Some recent reports include platooning at the hubs in a transportation network as a noncooperative coordination game~\cite{Johansson2023}, attacker-detector game for improving the security of platoons against cyber attacks~\cite{Basiri2019}, platoon formation as a coalitional game~\cite{Calvo2018}, and complete and incomplete information behavioral decision-making in a platoon using noncooperative game theory~\cite{Liu2023}. 

Differential games have been extensively used to address multi-robot systems formation control~\cite{Gu2008,LIN201454,Mylvaganam2017,Jond2019}. A platoon is a line formation. However, only a few research studies have utilized differential games for platooning. In~\cite{JondIET2023}, differential games for platooning under the predecessor-following and two-predecessor-following topologies for platoon control problems with vehicles governed by single integrator dynamics were solved analytically, and the closed-form expressions for the open-loop Nash equilibrium were derived. 

As the main contribution with respect to~\cite{JondIET2023}, this paper considers $i$) a linearized dynamics model that approximates the longitudinal dynamics of car-like vehicles and $ii$), collision avoidance. It is shown that a closed-form solution for the platoon control problem without collision avoidance in the context of a noncooperative differential game exists. Realizing that a closed-form solution for the game problem with collision avoidance is not available, we propose an estimated Nash strategy that is greatly simplified for implementation under an open-loop information structure. Both solutions' effectiveness is shown by the simulation studies. 

The paper is organized as follows. Section~\ref{formulation} presents a differential game model of the platoon control problem without collision avoidance. In Section~\ref{open-Nash}, we derive the open-loop Nash equilibrium and its associated opinion trajectories. The platoon control problem with collision avoidance is studied in Section~\ref{open-Nash-col}. In Section~\ref{simulation}, the results from previous sections are verified by simulations. Conclusions and future works are discussed in Section~\ref{conclusions}.

\section{Differential Game Formulation}\label{formulation}

We consider a homogeneous platoon of vehicles with the predecessor-following information topology as depicted in Fig.~\ref{fig:platoon}. Each vehicle follows its predecessor by maintaining a predefined fixed inter-vehicle distance using unidirectional information acquired directly from onboard sensors or via vehicle-to-vehicle connections in connected environments. The vehicles are equipped with cameras that detect their immediate preceding vehicle and laser scanners for measuring the distances. Suppose that there are $N+1$ vehicles in the platoon, indexed by $0$ through $N$ where $0$ corresponds to the lead vehicle and the rest to the following vehicles. The lead vehicle, or simply the leader, is at the front of the platoon and has a constant velocity. The following vehicles, or followers, adjust their control input to maintain their predefined distances from their predecessors.

\begin{figure}
\centering
\includegraphics[width=0.45\textwidth]{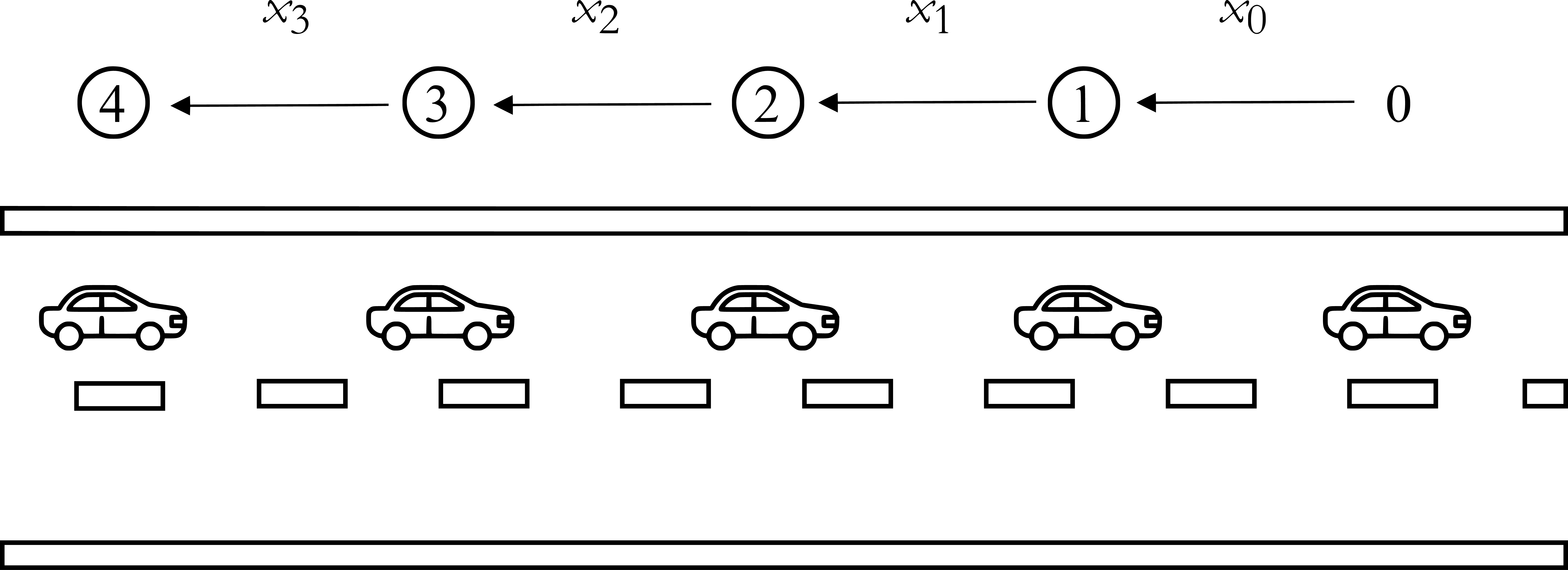}
\caption{A homogeneous vehicle platoon with predecessor-following topology. The corresponding information topology shows that each follower vehicle has the information of only its predecessor vehicle. } \label{fig:platoon}
\end{figure}

Vehicle dynamics is a nonlinear function of tire friction, rolling resistance, aerodynamic drag, gravitational force, the engine, the brake system, etc., fundamentally challenging theoretical analysis. Simplified nonlinear vehicle dynamics models have commonly been used to model vehicle longitudinal dynamics in a platoon. These nonlinear dynamics models govern the engine dynamics, brake system, and aerodynamic drag of each vehicle. By using the feedback linearization technique, as shown in~\cite{Xiao,Hu2020}, these nonlinear dynamics become linearized, which eases further theoretical analysis.

Let $p_i(t)$, $v_i(t)$, $a_i(t)$, and $u_i(t)$ denote the position, velocity, acceleration, and control input of the $i$th vehicle in the platoon, respectively. The engine time constant $\tau$ (also called the inertial time-lag) encompassed by the feedback linearized dynamics is, in reality, different even for identical vehicles. However, in this work, we consider a homogeneous platoon of follower vehicles with identical $\tau$. This assumption ensures that the platoon control problem that will be defined in this paper can be solved analytically with a closed-form solution. Each follower vehicle $i\in\{1,\ldots,N\}$'s feedback linearized dynamics is given by
\begin{align*}
    &    
\left\{
    	\begin{array}{lll}
		\dot p_i(t)=v_i(t)   
		&\\
		 \dot v_i(t)=a_i(t) 
		&\\
		 \tau\dot a_i(t)+a_i(t)=u_i(t) \\
	\end{array}
	\right.
\end{align*}
or, equivalently, in the following state-space form 
\begin{equation}\label{eq:dynamics}
    \dot{x}_i(t)=Ax_i(t)+Bu_i(t)
\end{equation}
where $x_i(t)=\begin{bmatrix}p_i(t)\\v_i(t)\\a_i(t)\end{bmatrix}$, $A=\begin{bmatrix}
0 & 1 & 0\\ 0 & 0 & 1\\0 & 0 & -\frac{1}{\tau}
\end{bmatrix}$, and $B=\begin{bmatrix}
0\\0\\\frac{1}{\tau}
\end{bmatrix}$. Note that the leader is supposed to move with constant velocity, i.e., $x_0=[p_0(t),v_0,0]^\top$, under the steady-state condition, i.e., $u_0(t)=0$. 

Follower vehicles try to maintain the inter-vehicular spacing $d_i$ between themselves and their immediate predecessors. Note that $d_i$ can be regarded as the length of each vehicle appended to it. The control objective is to ensure all the following vehicles reach consensual velocity with the lead vehicle while maintaining the predefined constant inter-vehicular spacing from their predecessors. In other words, the vehicle platoon shown in Fig.~\ref{fig:platoon} for any given bounded initial states $x_0(0)>x_1(0)>\cdots>x_N(0)$ achieves the desired platoon if control objectives 
\begin{align} \label{eq:objective}
   &\|x_{i-1}(T)-x_i(T)-\hat{d}_i\|^2\rightarrow 0
\end{align}
for all $i\in\{1,\ldots,N\}$ are satisfied for a sufficiently large finite time horizon $T$ and $\hat{d}_i=[d_i,0,0]^\top$. Note that although a homogeneous platoon of vehicles with identical dynamics is considered, the individual vehicles do not have to necessarily commit to an identical inter-vehicular spacing policy and can choose their own spacing policies according to their specifications.  

In the context of a differential game, the control objectives (\ref{eq:objective}) are transformed into PIs, each of which is supposed to be optimized by an individual vehicle in the platoon or a player in the game. The PI for each following vehicle $i\in\{1,\ldots,N\}$ is defined as 
\begin{align}\label{eq:cost}
    J_i=\omega_i\|x_{i-1}(T)-x_i(T)-\hat{d}_i\|^2+\int_0^T u_i^2(t)~\mathrm{dt}
\end{align}
where $\omega_i>0$ is a weighting parameter that penalizes the inter-vehicle displacement, and vehicles can adjust this parameter taking their personal interests or other personal factors such as their fuel amount in the tank into account. 

In this paper, we consider the platoon control problem (\ref{eq:dynamics}) and (\ref{eq:cost}) in the context of a noncooperative scenario differential game, where the notion of optimality is Nash equilibrium. A Nash equilibrium is a strategy combination of all players in a noncooperative game with the property that no one can gain a lower cost by unilaterally deviating from it.

\section{Open-Loop Nash Equilibrium}\label{open-Nash}

In the following theorem, we show the existence of a unique open-loop Nash equilibrium and then we present closed-form expressions for the equilibrium actions and their associated state trajectories for the underlying platoon control problem.

\begin{theorem}\label{theorem-Nash}
Consider a platoon of vehicles with the feedback linearized dynamics (\ref{eq:dynamics}) and PIs (\ref{eq:cost}). The platoon control problem as a noncooperative differential game admits a unique open-loop Nash equilibrium given by
    \begin{equation} \label{eq:notion-u}
        u_i(t)=-\sum_{j=1}^i\xi_j(t)
    \end{equation}
    where
    \begin{align}
        &\xi_i(t)=-\omega_iB^\top\mathrm{e}^{(T-t)A^\top}\left(I+\omega_i\Psi(T)\right)^{-1}\mathrm{e}^{TA}y_i(0),\label{eq:sol-xi}\\
        &\Psi(t)=\int_0^t\mathrm{e}^{(t-s)A}BB^\top\mathrm{e}^{(t-s)A^\top}\mathrm{ds}\label{eq:Psi}.
    \end{align}
The state trajectories associated with the equilibrium actions are given by
    \begin{equation} \label{eq:notion-x}
    x_i(t)=x_0(t)-\sum_{j=1}^i(y_j(t)+\hat{d}_j)
\end{equation}
where
\begin{align}\label{eq:traj}
    y_i(t)=\left(\mathrm{e}^{tA}-\omega_i\Psi(t) \left(I+\omega_i\Psi(T)\right)^{-1}\mathrm{e}^{TA}\right)y_i(0).
\end{align}
\end{theorem}

\begin{proof}
Let $y_i(t)=x_{i-1}(t)-x_i(t)-\hat{d}_i$ and $\xi_i(t)=u_{i-1}(t)-u_i(t)$ for all $i\in\{1,\ldots,N\}$ where (\ref{eq:notion-x}) and (\ref{eq:notion-u}) are easily verified, respectively. 

Vehicle dynamics (\ref{eq:dynamics}) is then expressed in terms of the new state vector $y_i(t)$ and new control input $\xi_i(t)$ as
\begin{equation}\label{eq:dynamics-y}
    \dot{y}_i(t)=Ay_i(t)+B\xi_i(t).
\end{equation}
Therefore, the platoon control problem as the noncooperative differential game (\ref{eq:dynamics}) and (\ref{eq:cost}) reduces to the following optimization
\begin{align*}
    \min_{\xi_i}\mathcal{J}_i=\omega_iy_i^\top(T) y_i(T)+\int_0^T\xi_i^2(t)~\mathrm{dt}
\end{align*}
subject to (\ref{eq:dynamics-y}).

Define the Hamiltonian for the above minimization 
\begin{align}\label{eq:non-coop-hamil}
    H_i=\xi_i^2(t)+\lambda_i^\top(t) (Ay_i(t)+B\xi_i(t))
\end{align}
for all $i\in\{1,\ldots,N\}$ where $\lambda_i(t)$ is the costate. According to Pontryagin’s minimum principle, the necessary conditions for optimality are $\frac{\partial H_i}{\partial \xi_i}=0$ and $\dot \lambda_i(t)=-\frac{\partial H_i}{\partial y_i}$. Applying the necessary conditions on (\ref{eq:non-coop-hamil}) yield
\begin{align} 
    &\xi_i(t)=-B^\top\lambda_i(t),\label{eq:neccesary-u}\\
    &\dot{\lambda}_i(t)=-A^\top\lambda_i(t) ,\quad  \lambda_i(T)=\omega_i y_i(T) \label{eq:neccesary-lam}
\end{align}
for $i\in\{1,\ldots,N\}$.

The solution of (\ref{eq:neccesary-lam}) is given by
\begin{align}\label{eq:sol-lam}
    \lambda_i(t)=\mathrm{e}^{(T-t)A^\top}\lambda_i(T)=\omega_i\mathrm{e}^{(T-t)A^\top}y_i(T).
\end{align}
Substituting (\ref{eq:neccesary-u}) into (\ref{eq:dynamics-y}) and using (\ref{eq:sol-lam}), we have
\begin{align*}
   \dot{y}_i(t)&=Ay_i(t)-BB^\top\lambda_i(t)\nonumber\\&=Ay_i(t)-\omega_iBB^\top\mathrm{e}^{(T-t)A^\top}y_i(T)
\end{align*}
where its solution is given by
\begin{equation}\label{eq:sol-y}
   y_i(t)=\mathrm{e}^{tA}y_i(0)-\omega_i\Psi(t) y_i(T) 
\end{equation}
where $\Psi(t)$ is defined in (\ref{eq:Psi}). Consider (\ref{eq:sol-y}) at $T$ as
\begin{equation}\label{eq:sol-yT}
   y_i(T)=\mathrm{e}^{TA}y_i(0)-\omega_i\Psi(T) y_i(T). 
\end{equation}
Equation (\ref{eq:sol-yT}) can be rewritten as
\begin{equation*}
   \left(I+\omega_i\Psi(T)\right) y_i(T)=\mathrm{e}^{TA}y_i(0) 
\end{equation*}
or
\begin{equation}\label{eq:sol-yTT}
   y_i(T)=\left(I+\omega_i\Psi(T)\right)^{-1}\mathrm{e}^{TA}y_i(0). 
\end{equation}

Note that $y_i(T)$ exists for every initial condition $y_i(0)$ iff $\left(I+\omega_i\Psi(T)\right)^{-1}$ exists. In other words, the game has an open-loop Nash equilibrium for every initial states $x_0(0),\cdots,x_N(0)$ iff (\ref{eq:sol-yTT}) can be calculated for any arbitrary final state $y_i(T)$ and accordingly, $x_i(T)$. If
so, the equilibrium actions are unique and exist for all $t\in[0,T]$. Otherwise, the game does not have a unique open-loop Nash equilibrium for every initial states $x_0(0),\cdots,x_N(0)$.

In the following, we show that the matrix $I+\omega_i\Psi(T)$ is invertible.

From (\ref{eq:Psi}), we have 
\begin{equation*}
    \mathrm{e}^{(t-s)A}BB^\top\mathrm{e}^{(t-s)A^\top}=\mathrm{e}^{(t-s)A}B(\mathrm{e}^{(t-s)A}B)^\top.
\end{equation*}
The product of any matrix and its transpose is always symmetric. Thus, $\omega_i\Psi(T)$ is symmetric. The matrix $\mathrm{e}^{(t-s)A}$ is positive definite and all its eigenvalues are positive. The matrix $BB^\top$ is a nonnegative diagonal matrix, and then the eigenvalues of the product of $\mathrm{e}^{(t-s)A}BB^\top\mathrm{e}^{(t-s)A^\top}$ still have nonnegative real parts. Therefore, all the eigenvalues of $I+\omega_i\Psi(T)$ in (\ref{eq:Psi}) have positive real parts. 

Substituting (\ref{eq:sol-yTT}) into (\ref{eq:sol-y}) and rearranging it, we obtain (\ref{eq:traj}). Similarly, substituting (\ref{eq:sol-lam}) into (\ref{eq:neccesary-u}) and then (\ref{eq:sol-yTT}) into it, we get (\ref{eq:sol-xi}). 

\end{proof}

\section{Collision Avoidance Estimated Nash Strategy}\label{open-Nash-col}
The platoon control problem (\ref{eq:dynamics}) and PIs (\ref{eq:cost}) and its solution in \textit{Theorem~\ref{theorem-Nash}} satisfy only the control objectives of maintaining a constant inter-vehicular spacing with the predecessor and maintaining the consensual velocities and accelerations of all followers with the leader. In addition, each following vehicle in the platoon has to ensure the crucial requirement of collision avoidance. 

Control designs that simultaneously guarantee the time-headway spacing and collision avoidance in platoons were the focus of a few reports~\cite{Lunze,Schwab}. The differential game literature on collision avoidance is from multi-robot systems~\cite{Mylvaganam2017,Cappello}.  

For the platoon control problem with collision avoidance, the PI for vehicle $i$ is redefined as 
\begin{align}\label{eq:cost-col}
    \hat{J}_i=J_i+\frac{1}{\mu_i\|x_{i-1}(T)-x_i(T)-\hat{r}_i\|^2+\epsilon}
\end{align}
for all $i\in\{1,\ldots,N\}$ where $\mu_i>0$ is a weighting parameter, $\epsilon>0$ is a positive scalar to ensure a non-zero denominator, and $\hat{r}_i=[r_i,0,0]^\top$ where $r_i$ is a safe distance from the predecessor for collision avoidance. If vehicle $i$ gets closer to its predecessor than $r_i$, a collision is unavoidable. 

The platoon control problem with collision avoidance in (\ref{eq:dynamics}) and (\ref{eq:cost-col}) is non-trivial and challenging to solve for its closed-form solution. We attempt to constitute an estimation of the exact solution that guarantees the collision avoidance behavior of followers.  

Define the following positive scalar function of $y_i(t)$
\begin{equation}\label{eq:func-f}
    f(y_i(t))=\frac{1}{\left(\mu_i(y_i(t)+\hat{d}_i-\hat{r}_i)^\top(y_i(t)+\hat{d}_i-\hat{r}_i)+\epsilon\right)^2}.
\end{equation}
Also, define
\begin{align}\label{eq:terminal-state-est}
   z_i(t)=&\Big(I+\big(\omega_i-\mu_if(\mathrm{e}^{tA}y_i(0))\big)\Psi(t)\Big)^{-1}\nonumber\\&\Big(\mathrm{e}^{tA}y_i(0)-\mu_if(\mathrm{e}^{tA}y_i(0))\Psi(t)(\hat{r}_i-\hat{d}_i)\Big)
\end{align}
where $z_i(T)=\hat{y}_i(T)$.

\begin{theorem} \label{theorem-col}
Consider a platoon of vehicles with the feedback linearized dynamics given in (\ref{eq:dynamics}) and PIs in (\ref{eq:cost-col}). Suppose that every vehicle $i$ utilizes the following estimation of its terminal state vector $y_i(T)$ from (\ref{eq:terminal-state-est}), i.e., $\hat{y}_i(T)$. For the platoon control problem with collision avoidance as a noncooperative differential game, the following estimations of the unique Nash equilibrium form collision-avoidance control inputs 
    \begin{equation}\label{eq:Nash-est}
        \hat{u}_i(t)=-\sum_{j=1}^i\hat{\xi}_j(t)
    \end{equation}
    where
    \begin{align}
        &\hat{\xi}_i(t)=-B^\top\mathrm{e}^{(T-t)A^\top}\times\nonumber\\&\Big(\big(\omega_i-\mu_if(\mathrm{e}^{TA}y_i(0))\big)\hat{y}_i(T)+\mu_if(\mathrm{e}^{TA}y_i(0))(\hat{r}_i-\hat{d}_i)\Big)\label{eq:sol-xi-col}.
    \end{align}
The state trajectories associated with the equilibrium actions are given by
    \begin{equation} \label{eq:Nash-traj-est}
    \hat{x}_i(t)=x_0(t)-\sum_{j=1}^i(\hat{y}_j(t)+\hat{d}_j)
\end{equation}
where
\begin{align}\label{eq:traj-col}
    &\hat{y}_i(t)=\mathrm{e}^{tA}y_i(0)-\Psi(t)\times\nonumber\\&\Big(\big(\omega_i-\mu_if(\mathrm{e}^{TA}y_i(0))\big)\hat{y}_i(T)+\mu_if(\mathrm{e}^{TA}y_i(0))(\hat{r}_i-\hat{d}_i)\Big) .
\end{align}
\end{theorem}

\begin{proof} The platoon control problem in (\ref{eq:dynamics}) and (\ref{eq:cost-col}) in terms of the state vector $y_i(t)$ and control input $\xi_i(t)$ reduces to the minimization of the following optimization
\begin{align*}
    \min_{\xi_i}\hat{\mathcal{J}}_i&=\mathcal{J}_i(\xi_i(t))+\nonumber\\&\frac{1}{\mu_i(y_i(T)+\hat{d}_i-\hat{r}_i)^\top(y_i(T)+\hat{d}_i-\hat{r}_i)+\epsilon}
\end{align*}
subject to (\ref{eq:dynamics-y}).

Define the Hamiltonian (\ref{eq:non-coop-hamil}) and by using the necessary conditions for optimality, we obtain (\ref{eq:neccesary-u}) and (\ref{eq:neccesary-lam}) with the following terminal condition
\begin{align}
    \lambda_i(T)=\omega_i y_i(T)-\mu_if(y_i(T))(y_i(T)+\hat{d}_i-\hat{r}_i) \label{eq:neccesary-lam-terminal}
\end{align}
for $i\in\{1,\ldots,N\}$.

The solution of (\ref{eq:neccesary-lam}) using the terminal condition (\ref{eq:neccesary-lam-terminal}) is given by
\begin{align}\label{eq:lamda-col}
    &\lambda_i(t)=\mathrm{e}^{(T-t)A^\top}\times&\nonumber\\&\Big(\big(\omega_i-\mu_if(y_i(T))\big)y_i(T)+\mu_if(y_i(T))(\hat{r}_i-\hat{d}_i)\Big).
\end{align}

Substituting (\ref{eq:lamda-col}), respectively, into (\ref{eq:neccesary-u}) and then into (\ref{eq:dynamics-y}), we have
\begin{align} 
    \xi&_i(t)=-B^\top\mathrm{e}^{(T-t)A^\top}\times\nonumber\\&\Big(\big(\omega_i-\mu_if(y_i(T))\big)y_i(T)+\mu_if(y_i(T))(\hat{r}_i-\hat{d}_i)\Big),\label{eq:xi-col}\\
    \dot{y}&_i(t)=Ay_i(t)-BB^\top\mathrm{e}^{(T-t)A^\top}\times\nonumber\\&\Big(\big(\omega_i-\mu_if(y_i(T))\big)y_i(T)+\mu_if(y_i(T))(\hat{r}_i-\hat{d}_i)\Big)\label{eq:dynamics-col}
\end{align}
for all $i\in\{1,\ldots,N\}$.

The solution of (\ref{eq:dynamics-col}) is given by
\begin{align*}
   y_i(t)&=\mathrm{e}^{tA}y_i(0)-\nonumber\\&\Psi(t)\Big(\big(\omega_i-\mu_if(y_i)\big)y_i(T)+\mu_if(y_i)(\hat{r}_i-\hat{d}_i)\Big) 
\end{align*}
where at $T$ it can be rearranged as the following
\begin{align*}
   \Big(I+\big(\omega_i-&\mu_if(y_i(T))\big)\Psi(T)\Big) y_i(T)=\nonumber\\&\mathrm{e}^{TA}y_i(0)-\mu_if(y_i(T))(\hat{r}_i-\hat{d}_i)\Psi(T) 
\end{align*}
or equivalently, 
\begin{align}\label{eq:traj-col-simp}
   y_i(T)=\Big(I+&\big(\omega_i-\mu_if(y_i(T))\big)\Psi(T)\Big)^{-1}\nonumber\\&\Big(\mathrm{e}^{TA}y_i(0)-\mu_if(y_i(T))(\hat{r}_i-\hat{d}_i)\Psi(T)\Big). 
\end{align}

It is obvious from (\ref{eq:traj-col-simp}) that every player $i$ (i.e., every vehicle in the platoon) for all $i\in\{1,\ldots,N\}$ requires the knowledge of $f(y_i(T))$ for every possible terminal state vector $y_i(T)$, which is too complex to acquire from the current expression. Therefore, control inputs $\xi_i(t)$ and their associated state trajectories $y_i(t)$ will not be available explicitly, and thus neither will the true Nash equilibrium $u_i(t)$ and its associated state trajectories $x_i(t)$. However, it is possible to obtain a simplified expression for $\xi_i(t)$ and $y_i(t)$ from (\ref{eq:traj-col-simp}) as follows.

Assume that every player $i$ utilizes $y_i(T)=\mathrm{e}^{TA}y_i(0)$ to calculate $f(y_i(T))$. Then we arrive at the terminal state estimation of $\hat{y}_i(T)$ from (\ref{eq:terminal-state-est}). Substituting $\hat{y}_i(T)$ into (\ref{eq:xi-col}) and (\ref{eq:dynamics-col}) we get the estimations of control inputs $\hat{\xi}_i(t)$ in (\ref{eq:sol-xi-col}) and their associated state trajectories $\hat{y}_i(t)$ in (\ref{eq:traj-col}), respectively. Therefore, the estimations of the unique Nash equilibrium actions and their associated state trajectories are given by (\ref{eq:Nash-est}) and (\ref{eq:Nash-traj-est}), respectively.   
\end{proof}

The proposed estimated solution approach is incapable of dealing with collision avoidance since it is designed to implement the collision avoidance behavior only at the horizon time $T$. To consider collision avoidance for all $t\in[0,T]$, the PIs (\ref{eq:cost-col}) must have the collision avoidance term inside the integration, which brings far more difficulty to designing an implementable solution.

To implement the estimated Nash strategy design approach to include collision avoidance for $t\in[0,T]$, we utilize the following solution
\begin{align} \label{eq:Nash-est-t}
        &\hat{\xi}_i(t)=-B^\top\mathrm{e}^{(T-t)A^\top}\times\nonumber\\&\Big(\big(\omega_i-\mu_if(\mathrm{e}^{tA}y_i(0))\big)z_i(t)+\mu_if(\mathrm{e}^{TA}y_i(0))(\hat{r}_i-\hat{d}_i)\Big)
\end{align}
and
\begin{align}\label{eq:Nash-traj-est-t}
    &\hat{y}_i(t)=\mathrm{e}^{tA}y_i(0)-\Psi(t)\times\nonumber\\&\Big(\big(\omega_i-\mu_if(\mathrm{e}^{tA}y_i(0))\big)z_i(t)+\mu_if(\mathrm{e}^{tA}y_i(0))(\hat{r}_i-\hat{d}_i)\Big) .
\end{align}
Note that the solution above still has an open-loop information structure since it consists only of the initial state vector and time. As it is seen from its definition in (\ref{eq:terminal-state-est}), the vector $z_i(t)$ is also a function of the initial state vector and time.

\section{Simulation Results}\label{simulation}

In this section, we provide simulation results to demonstrate the effectiveness of the proposed platoon control schemes in Section~\ref{open-Nash} and \ref{open-Nash-col}. Consider a homogeneous platoon of five vehicles, i.e., $N=4$, of which the front vehicle is the leader and is not subject to control, and the rest are the followers with their control inputs to be designed. The inertial time-lag parameter is arbitrarily selected as $\tau=0.5$. The initial states of the vehicles are
given by $x_0(0)= [23,2,0]^\top$, $x_1(0)=[18,2.5,1]^\top$, $x_2(0)=[11,3,1.5]^\top$, $x_3(0)=[6,1.5,0.8]^\top$, $x_4(0)=[1,2,1.2]^\top$. Suppose that vehicles in the desired platoon will be equally spaced by $d_1=d_2=d_3=d_4=2$. Also, let $r_1=r_2=r_3=r_4=1$, meaning that if  vehicle $i\in\{1,\ldots,4\}$ gets closer to its predecessor than $1$, a collision occurs. In the PIs, $\omega_1=6$, $\omega_2=3$, $\omega_3=8$, $\omega_4=5$, $\mu_1=12$, $\mu_2=10$, $\mu_3=1$, $\mu_4=5$,  and $T=10$ for the time interval of the game.

We first solve the platoon control problem without collision avoidance with the open-loop Nash strategy given by (\ref{eq:notion-u}) and its associated state trajectory (\ref{eq:notion-x}) in \textit{Theorem~\ref{theorem-Nash}}. Fig.~\ref{fig:Nash} shows the following vehicles achieving the desired platoon at the horizon time $T=10$. In addition, it also shows the longitudinal velocities and accelerations of the following vehicles reaching the lead vehicle's velocity and acceleration. However, a collision between the lead vehicle and vehicle 1 occurs approximately at $t=5$.  
\begin{figure*}[h]
    \centering
     \begin{minipage}[b]{1\textwidth}
    \includegraphics[width=0.24\textwidth]{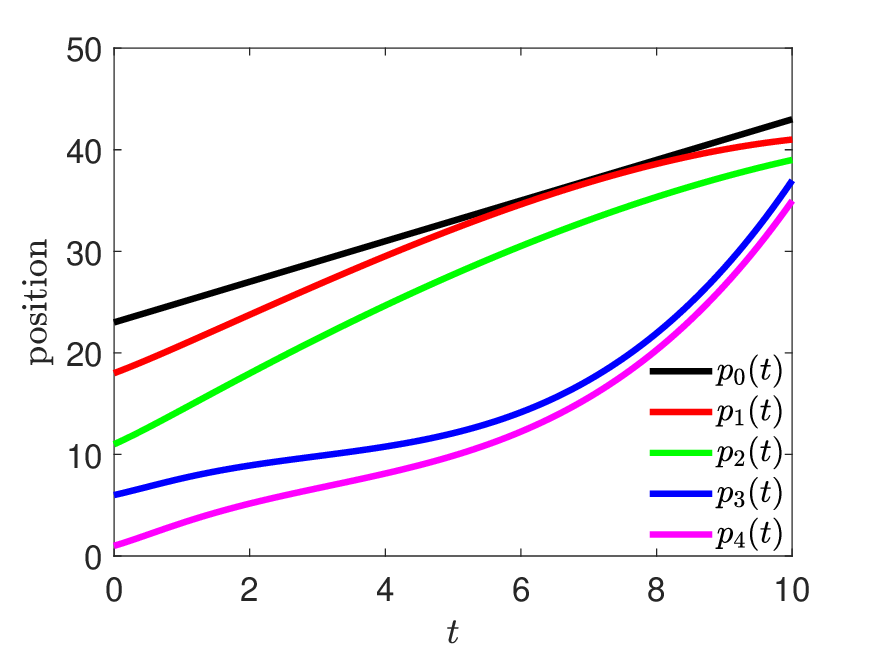}
    \includegraphics[width=0.24\textwidth]{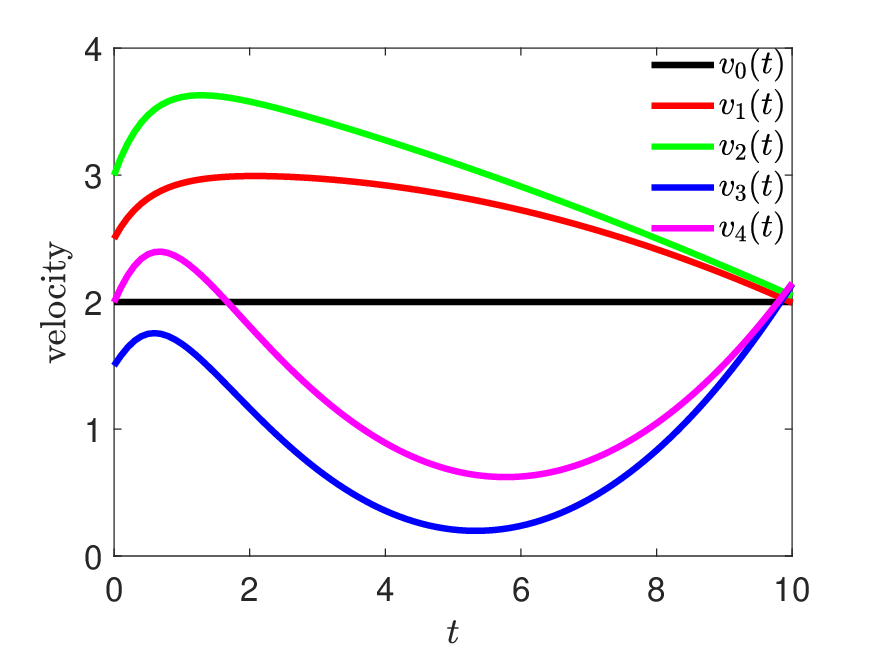}
    \includegraphics[width=0.24\textwidth]{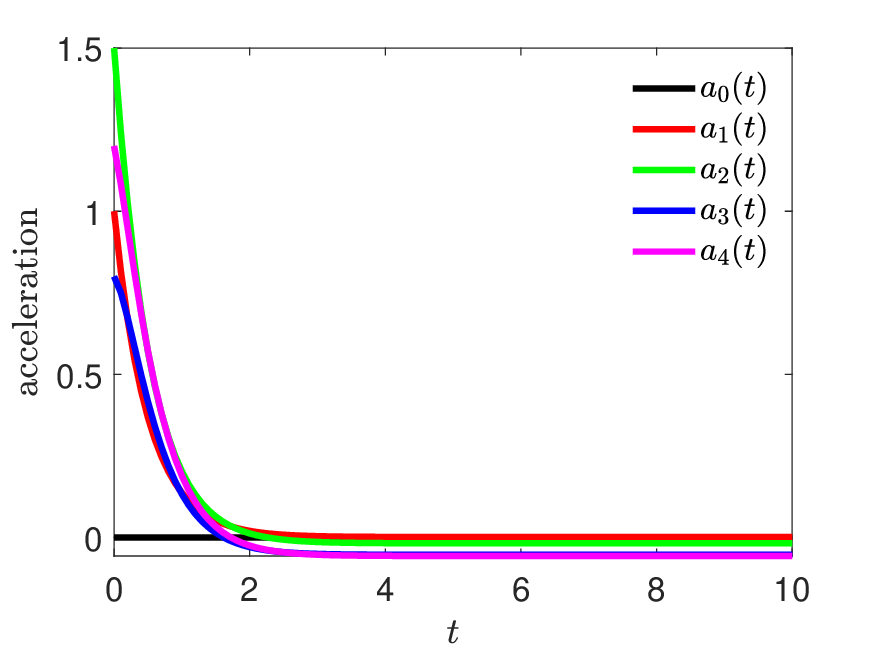}
    \includegraphics[width=0.24\textwidth]{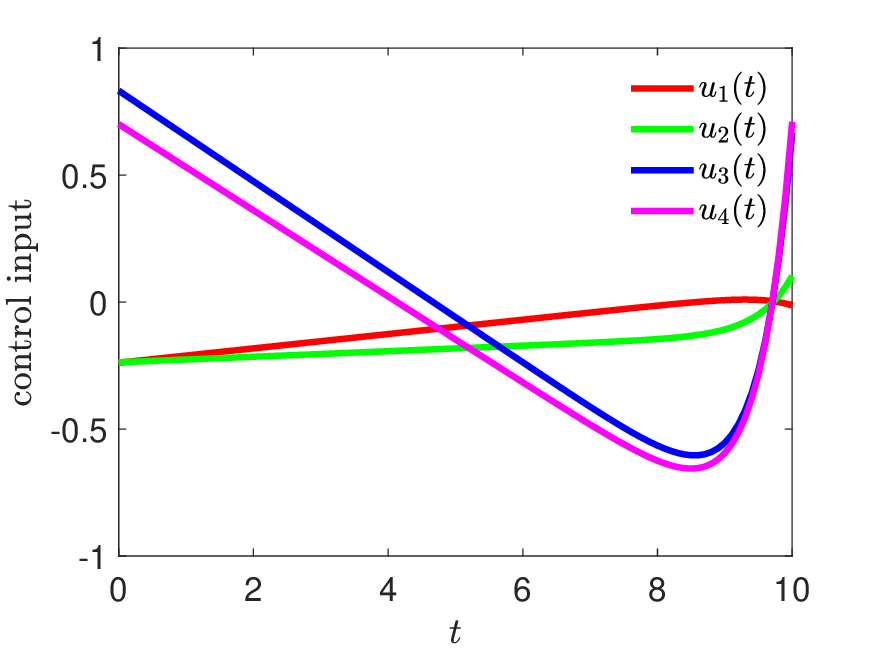}
    \end{minipage}
    \caption{Time histories of positions, velocities, accelerations, and control inputs of vehicles in the platoon control problem without collision avoidance. A collision between the lead vehicle and vehicle 1 is unavoidable.}
    \label{fig:Nash}
\end{figure*}

Once again, we resolve the platoon control problem with collision avoidance with the estimated Nash strategy solution (\ref{eq:Nash-est-t}) and (\ref{eq:Nash-traj-est-t}). The results in Fig.~\ref{fig:Nash-col} show the vehicles achieving basically the identical desired platoon at the horizon time $T=10$ as in the previous problem. However, it follows from Fig.~\ref{fig:Nash-col} that the collision between the lead vehicle and vehicle 1 is eradicated. Moreover, the vehicles achieve the desired platoon from approximately $t=3$ on, compared to the previous problem at horizon time $T=10$. This demonstrates the effectiveness of the proposed estimated solution implementation (\ref{eq:Nash-est-t}) and (\ref{eq:Nash-traj-est-t}) in coping with collision avoidance, early acquiring the desired platoon, and maintaining it. Note that the velocities of all follower vehicles should be equal to or greater than the leader's velocity (i.e., $v_i(t)\geq v_0$ for all $i\in\{1,\ldots,N\}$,) in order to avoid a collision risk~\cite{Lunze}. While this requirement is not seen in Fig.~\ref{fig:Nash} for the platoon control problem without collision avoidance, it is totally assured in Fig.~\ref{fig:Nash-col} for the platoon control problem with collision avoidance. A smaller velocity than the lead vehicle's velocity describes a braking maneuver for the following vehicle that creates a collision risk with its predecessor. Such maneuvers are seen for vehicles 3 and 4 in Fig.~\ref{fig:Nash}. Finally, the scalar function $f(\mathrm{e}^{tA}y_i(0))$ has been plotted for $t\in[0,T]$ in Fig.~\ref{fig:avoidance}. This function reveals collision risk for each following vehicle as a function of time and its initial state. It is seen that vehicles 1, 2, and 4 experience the peak collision risks at approximate times 4, 8, and 6, respectively, while vehicle 3 is not facing any collision risk the whole time.  
\begin{figure*}[h]
    \centering
    \begin{minipage}[b]{1\textwidth}
    \includegraphics[width=0.24\textwidth]{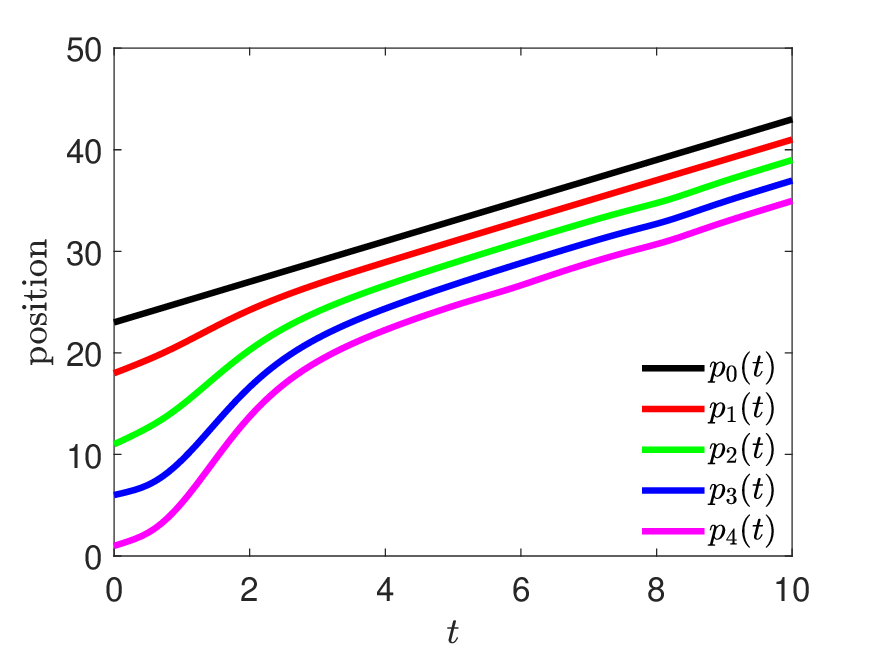}
    \includegraphics[width=0.24\textwidth]{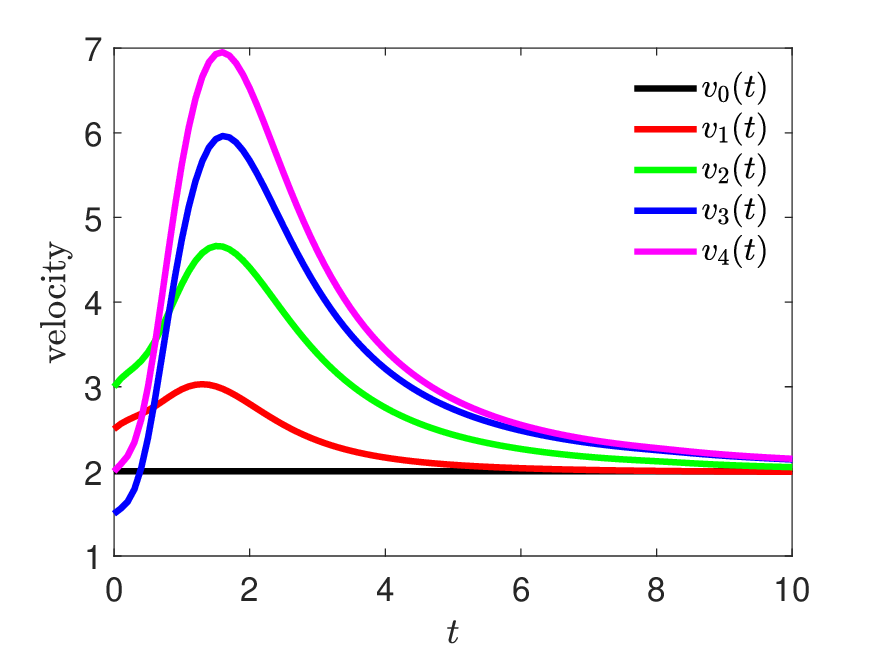}
    \includegraphics[width=0.24\textwidth]{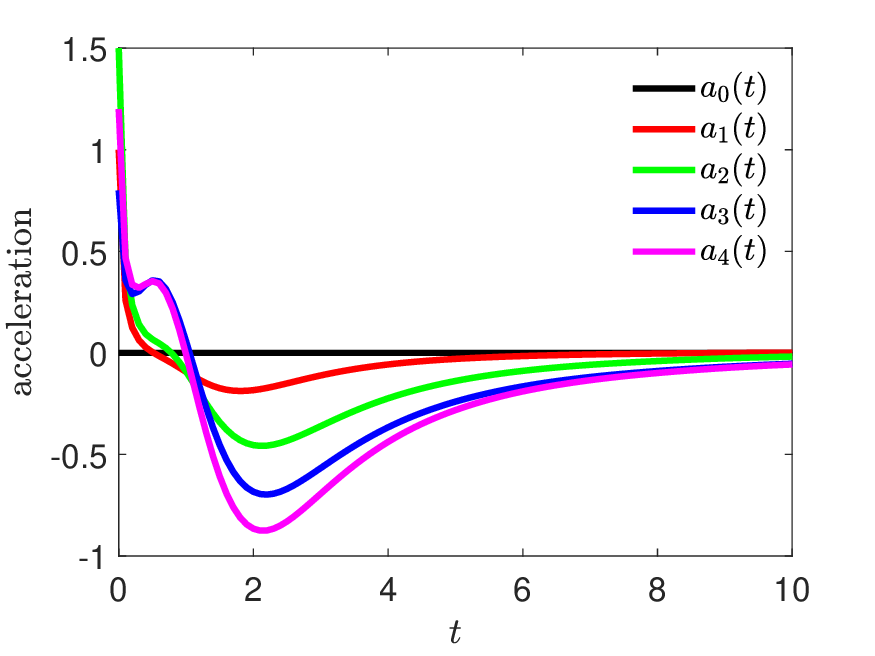}
    \includegraphics[width=0.24\textwidth]{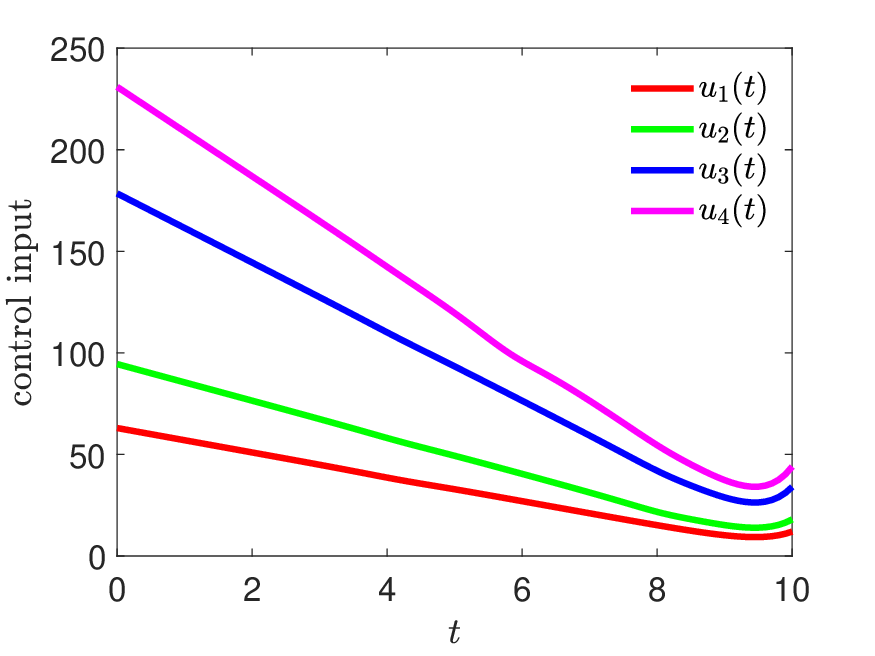}
    \end{minipage}
    \caption{Time histories of positions, velocities, accelerations, and control inputs of vehicles in the platoon control problem with collision avoidance. All following vehicles pursue collision-free trajectories.}
    \label{fig:Nash-col}
\end{figure*}
\begin{figure}[H]
\centering
       \includegraphics[width=0.25\textwidth]{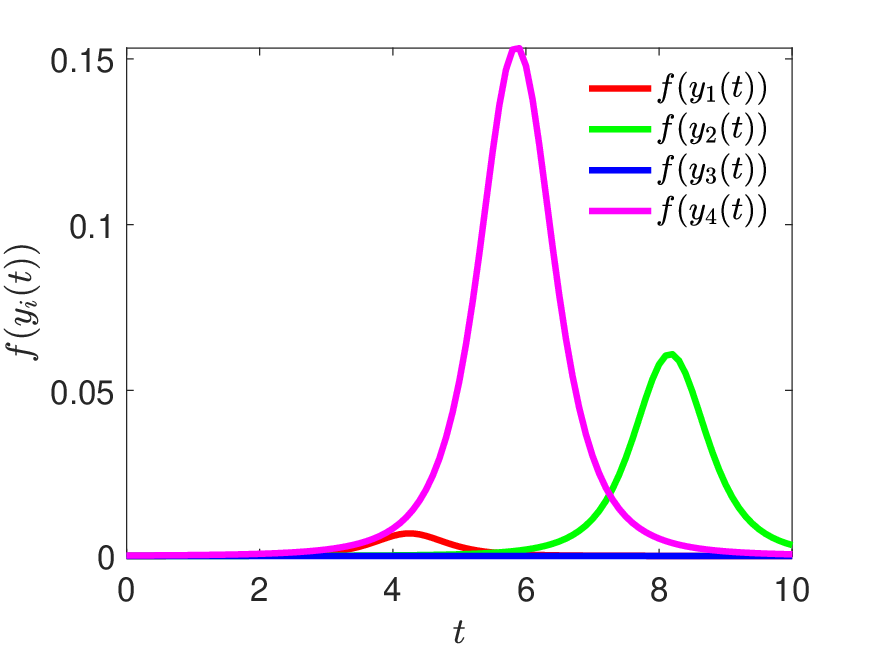}
\caption{Time histories of the scalar function $f(\mathrm{e}^{tA}y_i(0))$.}\label{fig:avoidance}
\end{figure}

\section{Conclusions}\label{conclusions}
In this paper, we have introduced differential game models for the homogeneous predecessor-following vehicle platoon control problem without and with collision avoidance. We obtained the closed-form expression for the unique Nash equilibrium and its associated state trajectories for the following vehicles. Simulation results have shown that the following vehicles acquire the desired platoon by committing to their self-enforcing controller based on Nash equilibrium actions. Furthermore, collision avoidance was considered in the platoon control problem, and the estimated Nash solution was proposed. The effectiveness of the proposed estimated Nash strategy design approach was observed in the simulation results. Future work will include more sophisticated platoon information topologies for connected environments, investigating feedback Nash equilibrium under feedback information differential games, and studying the string stability of the platoon.

%\addtolength{\textheight}{-17.9cm}   % This command serves to balance the column lengths
\section*{Acknowledgment}

This work was supported by SGS, V\v{S}B - Technical University of Ostrava, Czech Republic, under grant No. SP2023/012 “Parallel processing of Big Data X”.

\end{document}